\documentclass[conference]{IEEEtran}
\IEEEoverridecommandlockouts
\usepackage{cite}
\usepackage{amsmath,amssymb,amsfonts}
\usepackage{algorithmic}
\usepackage{graphicx}
\usepackage{subcaption}
\usepackage{textcomp}
\usepackage{xcolor}
\def\BibTeX{{\rm B\kern-.05em{\sc i\kern-.025em b}\kern-.08em
    T\kern-.1667em\lower.7ex\hbox{E}\kern-.125emX}}
    
\usepackage[utf8]{inputenc} % allow utf-8 input
\usepackage[T1]{fontenc}    % use 8-bit T1 fonts
\usepackage{hyperref}       % hyperlinks
\usepackage{url}            % simple URL typesetting
\usepackage{booktabs}       % professional-quality tables
\usepackage{amsfonts}       % blackboard math symbols
\usepackage{nicefrac}       % compact symbols for 1/2, etc.
\usepackage{microtype}      % microtypography
\usepackage{xcolor}         % colors
\usepackage{siunitx}

\usepackage{amsmath,bm}
\usepackage{amsthm}
\usepackage{mathtools}
\usepackage{centernot}
\usepackage{algorithm}
\usepackage{algorithmic}
\usepackage{comment}

\DeclareMathOperator*{\argmin}{arg\,min}

\newtheorem{proposition}{Proposition}

\allowdisplaybreaks

\begin{document}

\title{UAV Swarms for Joint Data Ferrying and Dynamic Cell Coverage via Optimal Transport Descent and Quadratic Assignment
}

\author{Kai Cui, Lars Baumgärtner, Burak Yilmaz, Mengguang Li, Christian Fabian, Benjamin Becker, \\Lin Xiang, Maximilian Bauer, Heinz Koeppl
\thanks{This work has been co-funded by the LOEWE initiative (Hesse, Germany) within the emergenCITY center, the State of Hesse and HOLM as part of the "Innovations in Logistics and Mobility" programme of the Hessian Ministry of Economics, Energy, Transport and Housing (HA project no.: 1010/21-12), and the Hessian Ministry of Science and the Arts (HMWK) within the projects "The Third Wave of Artificial Intelligence - 3AI" and hessian.AI.}
\thanks{The authors are with the Departments of Electrical Engineering and Information Technology, Computer Science, and Mechanical Engineering, Technische Universität Darmstadt, 64287 Darmstadt, Germany. (e-mail: {\tt\small  \{heinz.koeppl\}@tu-darmstadt.de}).}}

\maketitle

\begin{abstract}
Both data ferrying with disruption-tolerant networking (DTN) and mobile cellular base stations constitute important techniques for UAV-aided communication in situations of crises where standard communication infrastructure is unavailable. For optimal use of a limited number of UAVs, we propose providing both DTN and a cellular base station on each UAV. Here, DTN is used for large amounts of low-priority data, while capacity-constrained cell coverage remains reserved for emergency calls or command and control. We optimize cell coverage via a novel optimal transport-based formulation using alternating minimization, while for data ferrying we periodically deliver data between dynamic clusters by solving quadratic assignment problems. In our evaluation, we consider different scenarios with varying mobility models and a wide range of flight patterns. Overall, we tractably achieve optimal cell coverage under quality-of-service costs with DTN-based data ferrying, enabling large-scale deployment of UAV swarms for crisis communication.
\end{abstract}

\begin{IEEEkeywords}
UAV swarms, data ferrying, cell coverage, alternating minimization, quadratic assignment problem
\end{IEEEkeywords}

\section{Introduction}
In recent years, Unmanned Aerial Vehicles (UAV) have become essential tools for professional first responders thanks to their mobility and versatility. 
Today they are mainly used for taking aerial pictures of disaster areas and sensor readings autonomously. 
In our highly connected world, communication is vital for coordination of civilians, professional responders and even (partially) autonomous systems or IoT. 
Thus, when communication infrastructure is disrupted, alternative means of communication must be quickly established. 
For such challenging conditions where connectivity is intermittent or unavailable for long periods, Delay-Tolerant Networking (DTN) provides a commonly used way to enable resilient communication without stable end-to-end connectivity.
In this store-carry-and-forward approach, data is stored at intermediate nodes and forwarded to the next node until connectivity is available.
Thus, message delivery depends on the mobility of the participating nodes, which act as "data mules" to physically transfer the data.
UAVs can be rapidly deployed to affected areas during emergencies such as natural disasters.
By carrying wireless communication equipment, UAVs can provide temporary communication coverage to DTNs and data ferrying in remote or disaster-stricken areas.
UAVs can also fly to areas difficult to access, such as mountainous or forested areas, or survey areas and collect data useful for disaster response. 

Even though researchers have proposed UAVs as data ferries in the past~\cite{zobel_topology-aware_2020, udroiu_data_2021, arafat_location-aided_2018, yoon_adaptive_2017}, there remain many open challenges for practical deployments.
Finding optimal coverage of the nodes on the ground is essential. 
If the locations of ground nodes are not known, flight patterns must either be optimized for area coverage, or a sweep of the area must be performed to detect nodes before providing coverage. 
Furthermore, nodes on the ground might also be moving, and different radio link technologies have different characteristics such as communication range, bandwidth, etc. 
Limited resources available on the UAV, e.g., data storage capacities and battery power, are major challenges to take into account when planning flight trajectories and missions.
Finally, a swarm of UAVs can increase the communication capabilities significantly, but requires additional effort to coordinate and cooperate.
Here, we address the optimal coverage of nodes in a tractable manner, amenable to large swarms and dynamically moving user nodes.

At the same time, another approach to communication via UAVs is based on the usage of UAVs as mobile cellular base stations. With the miniaturization of base station equipment, exploiting UAVs as aerial base stations to reinforce or complement cellular coverage has attracted significant interests in academia, industry, and standardization organizations such as the 3rd Generation Partnership Project (3GPP) \cite{3GPPstd}. Unlike the DTN scenario, aerial base stations are designed to provide reliable, energy-efficient cellular communication services with guarantees for e.g. communication data rates, reliability, delays, and information timeliness, even in challenging emergency situations. 
While aerial base stations can enable flexible deployments and relocation, establish strong line-of-sight (LoS) channels, and enhance resilience of ground networks amid malfunction, they also need to overcome several new challenges such as time-varying topologies, blockages in urban environments, and LoS interference \cite{Xiang20wcnc}. Moreover, UAVs are often subject to limited size, weight, power, and radio resources.

In order to deal with resource constraints in crisis scenarios, it is thus of importance to consider not only DTN-based communication, but also the availability of dynamic cellular coverage as a second layer of communication infrastructure reserved for high-priority data such as emergency calls or command and control of dynamically moving users on the ground. Therefore, it is important to enable both data ferrying and cellular coverage simultaneously, as depicted in Fig.~\ref{fig:overview}. In this work, for cell coverage, we generalize optimal transport formulation for optimal UAV swarm cell association \cite{mozaffari2017optimal} to include UAV capacities that need not be fully assigned to all users. Furthermore, in contrast to previous work, we then solve both coverage and data ferrying jointly. The latter is achieved by first solving the coverage problem via alternating minimization for dynamically tracking capacity-constrained clusters, and then applying combinatorial optimization problems such as the traveling salesman problem (TSP) and related quadratic assignment problems (QAP) for inter-cluster data ferrying.

\begin{figure}
    \centering
    \includegraphics[width=0.9\linewidth]{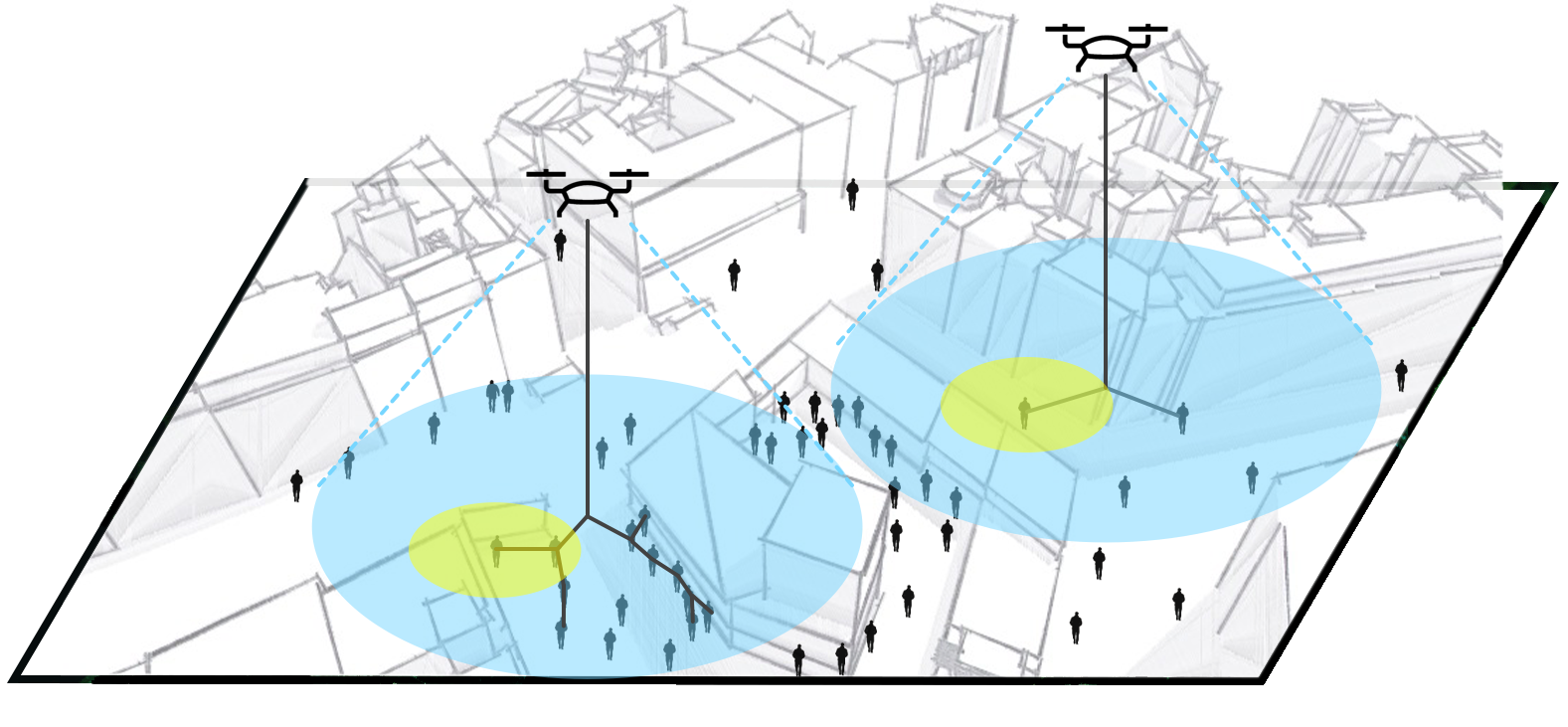}
    \caption{Overview of the proposed two-layer communication infrastructure. Blue ellipses represent the dynamic cell coverage provided by the UAVs. In contrast, grey lines indicate exemplary connections between DTN nodes, with yellow ellipses representing DTN communication range.}
    \label{fig:overview} 
\end{figure}

Our contributions can be summarized as: (i) Formulating an optimal-transport-based problem and algorithm for optimal cell coverage via UAVs with novel capacity constraints; (ii) Enabling distance-optimized and delay-optimized data ferrying by transporting data between user clusters via combinatorial optimization problems; (iii) Combining data ferrying with cell coverage in order to optimally use available UAVs for both DTN and cellular base stations, giving extensive theoretical and empirical support for the algorithm by considering different scenarios with varying mobility models and a wide range of flight patterns. To the best of our knowledge, our work constitutes one of the first to combine data ferrying with cell coverage, with applicability in resource-constrained scenarios or crises. We begin by describing the underlying scenario and our proposed algorithms. We then move on to a theoretical and empirical evaluation, including a small real demonstration over a city model using indoor UAVs. Lastly, we close by discussing related and future work.
\section{Background}
In this section, we provide a brief background on the scenario and model on which we base our algorithms. Precise values for variables used in evaluation are given in Sec.~\ref{sec:eval}.

Consider a potentially large number of users on the ground, and many UAVs for data ferrying and cell coverage with limited capacity, such that each UAV may serve only a limited number of users. We assume that at any time $t \in \mathbb N$ there are users $i = 1, \ldots, M$ at time-dependent positions $X_i(t) = (x_i(t), y_i(t)) \in \mathcal X$ in the area of operations $\mathcal X \subseteq \mathbb R^2$, following a certain mobility model as will be discussed in Sec.~\ref{sec:mobility}. This can be understood e.g. as mobility data of a city's mobile phone users over the day, or mobility behavior obtained from surveying in crisis situations. Analogously, for each UAV $j = 1, \ldots, N$ we define UAV positions $\xi_j(t) \in \mathcal X$, which can be moved at a maximum velocity $v$ in order to achieve dynamic user coverage and data ferrying. We write $\mathbf X(t)$ for the vector of all user locations, and similarly $\boldsymbol \xi(t)$ for UAV locations.

\subsection{Delay-tolerant networking}
Delay- or disruption-tolerant networking (DTN) uses a \textit{Store, Carry \& Forward architecture} to cope with intermittent network connectivity. 
Data is transmitted as bundles on a hop-to-hop basis, thus, no stable end-to-end route between a source and a destination is needed.
Instead, intermediate nodes act as data mules, physically carrying data around until it reaches its destination.
Therefore, besides routing strategies such as epidemic flooding and PRoPHET \cite{lindgren2003probabilistic}, node mobility plays an essential role in bundle delivery rates for such networks.
The official standard protocol as defined by the \textit{IETF DTN Working Group} is the \textit{Bundle Protocol (BP)} version 7 (RFC 9171)~\cite{RFC9171}. 

The first goal of our work is thus to use UAVs as controllable mobile nodes to improve the DTN message delivery rates and delays. In our work, we consider epidemic flooding via WiFi, using the TheONE simulator \cite{keranen2009one}, where we generate one message every $T_{\mathrm{DTN}} > 0$ time units at a random ground user, with another random ground user as destination.

\subsection{Cell coverage}
In contrast to DTN, cell coverage defines the problem of covering all users optimally at any time, using UAVs as mobile base stations. For example, one could provide LTE coverage via UAVs \cite{lin2018sky} for uplink and downlink data communication to a limited number of ground users within range $r > 0$.

To solve the problem of providing optimal cell coverage, we assume that the underlying cell association problem of associating ground users to UAVs at any specific time instant is solved optimally through an optimal transport formulation. More specifically, we conceptually generalize the optimal transport procedure introduced in \cite{mozaffari2017optimal} by allowing UAV base stations to have arbitrary constraints on total capacity, i.e. not all UAVs must be fully assigned to users, and vice versa. 

\subsubsection{Optimal cell association}
Using a cost function $d$ such as distance on $\mathcal X$ to define costs $d(X, \xi)$ of associating users at position $X$ with a UAV at position $\xi$, we define the optimal transport cost $W_d \left( \frac 1 M \sum_{i=1}^M\delta_{X_i(t)}, \frac 1 N \sum_{j=1}^N \delta_{\xi_j(t)} \right)$ (see e.g. \cite{villani2009optimal}) between the locations (empirical distribution) of users $\frac 1 M \sum_{i=1}^M\delta_{X_i(t)}$ and UAVs $\frac 1 N \sum_{j=1}^N \delta_{\xi_j(t)}$. Choosing
\begin{align}
    d(X, \xi) = \mathbf 1_{(r, \infty)}(\lVert X - \xi \rVert_2),
\end{align}
allows the optimal transport cost to formalize the notion of best achievable coverage under optimal cell association \cite{mozaffari2017optimal} for a maximum communication range of $r > 0$, assuming that the total capacity of all UAVs is equal to the total number of users on the ground, i.e. all UAVs must be fully assigned to users.

\subsubsection{Capacity constraints}
So far, the optimal transport formulation assumes the same amount of mass between UAVs and users, i.e. all UAVs must be fully assigned to users, and vice versa. To model differing capacity cases, we formally embed the user positions in the extended space $\bar{\mathcal X} \coloneqq \mathcal X \times \{0, 1\}$ by $\bar X(t) \coloneqq (X(t)^T, 0)^T$, and similarly for UAVs $\bar \xi(t) \coloneqq (\xi(t)^T, 0)^T$. We then instead use
\begin{align}
    d(\bar X, \bar \xi) = \mathbf 1_{(r, \infty)}(\lVert \bar X - \bar \xi \rVert_2) \cdot \mathbf 1_{\{1\}}(\bar X_3) \cdot \mathbf 1_{\{1\}}(\bar \xi_3),
\end{align}
where $\cdot_3$ denotes the last component (zero for real UAVs and users), to allow for UAVs to not be fully assigned to users (in case of more capacity than needed by users): This is done by adding virtual users at $(0, 0, 1) \in \bar{\mathcal X}$ to obtain the new cost
\begin{multline} \label{eq:tcost}
    c(\mathbf {\bar X}(t), \bar{\boldsymbol \xi}(t)) = W^{\mathrm{vec}}_d(\mathbf {\bar X}(t), \bar{\boldsymbol \xi}(t)) \coloneqq \\ 
    W_d \left(\frac 1 {MR} \sum_{i=1}^M \delta_{\bar X_i(t)} + (1 - \frac{1}{R}) \delta_{(0, 0, 1)}, \frac 1 N \sum_{j=1}^N \delta_{\bar \xi_j(t)} \right)
\end{multline}
which formalizes the best achievable coverage for higher capacity of UAVs than required by users, i.e. if more capacity is available than required by users, the unused UAV capacity is assigned to a virtual user at zero cost. Here we add a mass of $(1 - \frac{1}{R})$ virtual users, where $R > 1$ denotes the ratio between total UAV capacity and total ground users. An analogous argument with virtual UAVs and $R < 0$ allows for the case with less total UAV capacity than required by users.

\subsubsection{Quality of service (QoS)}
Lastly, in practice, another point of consideration is quality of service (QoS), e.g. the quality of downlink or uplink communication. We may optimize e.g. the sum of both the number of covered users and their QoS, using
\begin{multline}
    d_{QoS}(\bar X, \bar \xi) = \mathbf 1_{(r, \infty)}(\lVert \bar X - \bar \xi \rVert_2) \cdot \mathbf 1_{\{1\}}(\bar X_3) \cdot \mathbf 1_{\{1\}}(\bar \xi_3) \\
    \quad + \lVert \bar X - \bar \xi \rVert_2 \cdot \mathbf 1_{[0, r]}(\lVert \bar X - \bar \xi \rVert_2)
\end{multline}
where QoS cost increases linearly with distance, i.e. QoS decays linearly. In practice, one could use e.g. a signal-to-noise ratio.

\section{Algorithm design}
Next, we discuss the approach used in our work to achieve scalable algorithms with decentralized deployment capability for (i) maximizing cell coverage, and (ii) ferrying data between cell clusters in synchronous updates.

\subsection{Optimal cellular coverage}
Our model provides a mathematical foundation to the cell coverage problem where UAVs have a limited capacity of users. What remains is to find optimal UAV locations $\boldsymbol \xi(t)$ at any time $t$. We propose an iterative assignment method which finds local optima with theoretical convergence guarantees of the otherwise hard problem. More specifically, we apply alternating minimization (AM) on an optimal transport formulation of the coverage problem, generalizing algorithms for the NP-hard $k$-means \cite{hartigan1975clustering} or $k$-medoids \cite{kariv1979algorithmic} problem in clustering, by limiting cluster sizes. A simple example is illustrated in Fig.~\ref{fig:example}.

\begin{figure}
    \centering
    \hfill{}
    \includegraphics[width=0.99\linewidth]{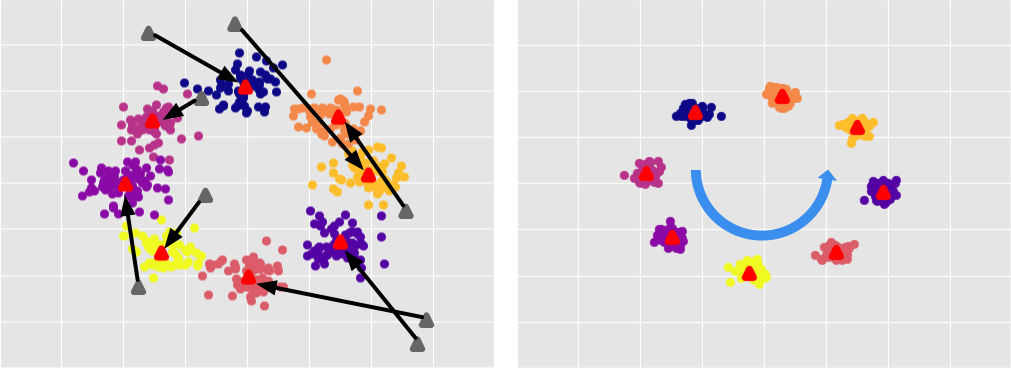}
    \hfill{}
    \caption{Simple example of 8 UAVs tracking 8 clusters of users (colored dots). UAVs start at random initial  starting points (grey triangles). Black arrows depict the moving trajectories of each UAV (red triangle). As the user clusters move counterclockwise, the UAVs successfully track the dynamic motion of the clusters.}
    \label{fig:example}
\end{figure}

We find optimal positions of UAVs $\boldsymbol \xi$ for given current locations of users, such that when the UAVs track these positions over time, the cell coverage of users is optimized. Assume a set of $M'$ users $\mathbf {\bar X} \in \bar{\mathcal X}^{M'}$, including virtual users. For any UAV (or cluster) $i$, we keep assignments $C_i \subseteq [M']$ of users to a UAV with $\mathbf C = (C_i)_{i \in [N]}$. If $\mathbf C$ constitutes the currently optimal assignment of users to UAVs (an optimal transport plan \cite{villani2009optimal}) given that the UAVs are located at $\boldsymbol {\bar \xi}$, then the optimal transport cost \eqref{eq:tcost} is given by
\begin{align} \label{eq:Jcost}
    J_{\mathbf X(t)}(\mathbf C, \boldsymbol \xi) \coloneqq \sum_{i \in [N]} \sum_{j \in C_i} d(\bar \xi_i, \bar X_j) = W^{\mathrm{vec}}_d(\mathbf {\bar X}(t), \bar{\boldsymbol \xi}(t)).
\end{align}

Therefore, we need to minimize \eqref{eq:Jcost} over all drone locations $\boldsymbol \xi$ and associated optimal transport plans $\mathbf C$. For this purpose, we apply an AM algorithm to minimize the optimal transport coverage problem by iterating
\begin{align}
    \mathbf C^{(n+1)} = \argmin_{\mathbf C} J_{\mathbf X(t)}(\mathbf C^{(n)}, \boldsymbol \xi^{(n)}) \label{eq:alt1}\\
    \boldsymbol \xi^{(n+1)} = \argmin_{\boldsymbol \xi \in \mathcal X^N} J_{\mathbf X(t)}(\mathbf C^{(n+1)}, \boldsymbol \xi^{(n)}) \label{eq:alt2}
\end{align}
over iterations $n$. The algorithm monotonically improves \eqref{eq:Jcost} and is thus guaranteed to converge. In other words, our algorithm repeatedly computes optimal transport solutions in \eqref{eq:alt1}, and reassigns locations of UAVs in \eqref{eq:alt2} to monotonically improve the current assignment of locations, which gives us at any time $t$ a locally optimal set of UAV locations. 

In the first step \eqref{eq:alt1}, we compute optimal transport plans $\mathbf C$ for fixed $\boldsymbol \xi$ via the POT library \cite{flamary2021pot}, which uses a linear program formulation as in \cite{bonneel2011displacement}. However, the exact computation of minimal $\boldsymbol \xi^{(n+1)}$ in the second step \eqref{eq:alt2} is out of reach via a facility location problem \cite{carlo2012capacitated}, or a Wasserstein barycenter problem \cite{borgwardt2021computational} even in the limiting relaxation of infinitely many UAVs (i.e. a large mean field UAV swarm). Hence, we instead use the medoids approach of choosing cluster centroids from a restricted set of user positions
\begin{align}
    \boldsymbol \xi^{(n+1)} = \argmin_{\boldsymbol \xi \subseteq X(t)} J_{\mathbf X(t)}(\mathbf C, \boldsymbol \xi) \label{eq:alt2p}
\end{align}
where we write $\boldsymbol \xi \subseteq X(t)$ whenever for all $i = 1, \ldots, N$ we have $\xi_i \in \{X_1(t), \ldots, X_M(t)\}$. The advantages of this approach are that we can easily compute the minimum by summing over distances and choosing the minimum
\begin{align}
    \argmin_{\boldsymbol \xi \subseteq X(t)} J_{\mathbf X(t)}(\mathbf C, \boldsymbol \xi) = \bigtimes_{i=1}^N \argmin_{\xi_i \in X(t)} \sum_{j \in C_i} d(\bar \xi_i, \bar X_j),
\end{align}
and that it ensures user-to-UAV connectivity also for DTN purposes by tracking one ground user's position via UAVs.

\begin{algorithm}[b]
    \caption{AM for optimizing centroid positions}
    \label{alg:am}
    \begin{algorithmic}[1]
        \STATE Input: Initial centroids $\boldsymbol \xi^{(0)}$.
        \FOR {iterations $n=1, \ldots, N_{\mathrm{AM}}$}
            \STATE Compute $\mathbf C^{(n)} = \argmin_{\mathbf C} J_{\mathbf X(t)}(\mathbf C^{(n-1)}, \boldsymbol \xi^{(n-1)})$.
            \STATE Compute $\boldsymbol \xi^{(n)} = \argmin_{\boldsymbol \xi \in \mathcal X^N} J_{\mathbf X(t)}(\mathbf C^{(n)}, \boldsymbol \xi^{(n-1)})$.
        \ENDFOR
        \RETURN Optimal centroid locations $\boldsymbol \xi^{(N_{\mathrm{AM}})}$
    \end{algorithmic}
\end{algorithm}

As a side remark, in order to avoid oscillating assignment of users that are out of reach to different clusters, in the algorithm we use a slightly adapted, "leaky" QoS distance measure
\begin{multline}
    \hat d(\bar X, \bar \xi) = \mathbf 1_{(r, \infty)}(\lVert \bar X - \bar \xi \rVert_2) \cdot \mathbf 1_{\{1\}}(\bar X_3) \cdot \mathbf 1_{\{1\}}(\bar \xi_3) \\
    \cdot (r + 0.01 \lVert \bar X - \bar \xi \rVert_2)
    + \lVert \bar X - \bar \xi \rVert_2 \cdot \mathbf 1_{[0, r]}(\lVert \bar X - \bar \xi \rVert_2)
\end{multline}
that is not constant on $\lVert \bar X - \bar \xi \rVert_2 > r$, to ensure cluster stability.

We assign UAVs to track the desired cluster locations. We update clusters dynamically by rerunning the algorithm after each time interval of $\Delta T$ time units. Warm starting with the last centroid positions keeps the cluster locations stable over time even if users are moving, allowing for UAVs to consistently track a moving cluster. This is important, since the users on the ground move over time (Fig.~\ref{fig:example}), and we would avoid clusters to be bound to specific users or locations. Instead, clusters dynamically readjust as underlying users move, allowing for certain users to switch from one cluster to another, or for the centroid to move with the users.

\subsection{Optimal inter-cluster DTN data ferrying}
All that remains is to rotate UAVs over clusters to facilitate data ferrying between clusters. Assuming that DTN messages are spread among users of any single cell cluster, the goal is to transport data between all clusters. We realize this diffusion of data by fixing an update interval $T_{\mathrm{rot}} > 0$ after which all UAVs change from their current cluster to another. Here, sufficiently high $T_{\mathrm{rot}}$ allows for good coverage during tracking of clusters by UAVs, and each rotation of UAV-to-cluster assignments also allows for a good spread of DTN messages. We introduce two formulations for optimization of energy and DTN latency.

\begin{figure}
    \centering
    \hfill{}
    \includegraphics[width=0.8\linewidth]{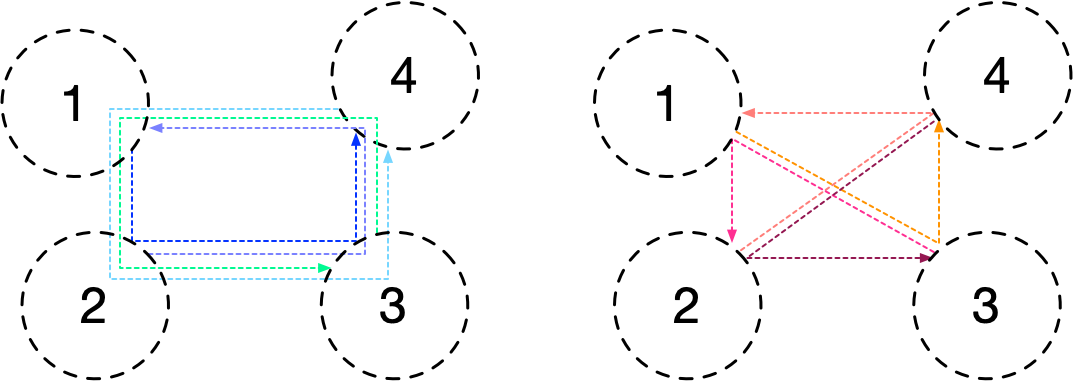}
    \hfill{}
    \caption{Two optimization formulations of data ferrying among clusters $1$ to $4$, and the resulting $4$ UAV routes. The blue lines (left) indicate TSP-based rotation with $3$ jumps per UAV, and the orange lines (right) indicate binary-jumping rotation with $2$ jumps per UAV.}
    \label{fig:modes}
\end{figure}

\subsubsection{TSP-based rotation}
Given $T_{\mathrm{rot}}$, a simple solution that minimizes energy costs associated with the total distance travelled by UAVs, and still ensures that all UAVs eventually reach all clusters, is given by the solution to a TSP. We find the shortest cycle in terms of distance travelled at the time of computation, and let all UAVs rotate through the cycle to eventually spread all messages between all clusters. 

As a result, if users are static or slow in comparison to the update time $T_{\mathrm{rot}} > 0$, the total distance travelled by UAVs is
\begin{align} \label{eq:d1}
    D_{\mathrm{tot}} = (N-1) \langle \mathbf D, \mathbf W \rangle_F = (N-1) \operatorname{tr}( \mathbf D \mathbf W ),
\end{align}
assuming that any $i$-th and $(i+1)$-th clusters are connected on the cycle (see also Fig.~\ref{fig:modes}), where $\mathbf D$ is the current distance matrix between cluster centroids $D_{ij} = d(\xi_i, \xi_j)$, $\mathbf W$ is the weight matrix with entries $w_{ij} = 1$ whenever $j = i + 1 \mod N$, and $\langle \cdot, \cdot \rangle_F$ denotes the Frobenius inner product defined by elementwise multiplication and summing of all entries.

By considering any possible permutation of cluster indices, the TSP can then briefly be written as the QAP \cite{cela2013quadratic}
\begin{align}
    \min_{\pi \in \mathrm{Sym}([N])} \operatorname{tr}( \mathbf D \mathbf P_\pi \mathbf W \mathbf P_\pi^T)
\end{align}
where $\pi$ is any permutation of $[N] \coloneqq \{1, \ldots, N\}$ from the symmetry group $\mathrm{Sym}$ of $[N]$, and $\mathbf P_\pi$ its corresponding permutation matrix. The TSP can then be solved using any standard method, for which we use OR-Tools \cite{ortools}.

An advantage of this method is that the distances travelled by UAVs during rotation are minimized, i.e. the method saves energy. A disadvantage is its slow diffusion of information between clusters, as it takes up to $N-1$ rotations to transport an existing message from one cluster to another. A variant of this method is that we can even use less UAVs than required to cover all clusters by assigning each UAV to track a cluster, and filling the other clusters with virtual UAVs. However, a disadvantage of using less UAVs than clusters is that full coverage is no longer achieved, making this solution suitable only if not enough UAVs are available to achieve full coverage of users. The full algorithm is given in Algorithm~\ref{alg:full}.

\begin{algorithm}[b]
    \caption{Joint data ferrying and cell coverage}
    \label{alg:full}
    \begin{algorithmic}[1]
        \STATE Initialize cluster centroids $\boldsymbol \xi^*$ and UAV-cluster assignments $\alpha \in \mathrm{Sym}([N])$.
        \FOR {each epoch $n=0, 1, \ldots, T_{\mathrm{tot}}/\Delta T$}
            \STATE Run Algorithm~\ref{alg:am} to obtain new centroids $\boldsymbol \xi^*$ (initialize with current $\boldsymbol \xi^*$).
            \IF{Rotation of UAVs has finished} 
                \STATE Compute new solution permutation $\pi$ to QAP using OR-Tools \cite{ortools} or Algorithm~\ref{alg:ga}. 
            \ENDIF
            \IF{Rotation update time $T_{\mathrm{rot}}$ has passed} 
                \STATE For TSP: Rotate UAVs by $\alpha = \alpha + 1 \mod N$.
                \STATE For binary-jumping: Rotate UAVs by $\alpha = \alpha + 2^{\lceil \log N \rceil - (n \mod \lceil \log N \rceil)} \mod N$.
            \ENDIF
            \FOR {each UAV $i=1,\ldots,N$}
                \STATE Command to move UAV $i$ towards cluster $\xi^*_{\pi(\alpha(i))}$.
            \ENDFOR
            \STATE Pass $\Delta T$ time units.
        \ENDFOR
    \end{algorithmic}
\end{algorithm}

\subsubsection{Binary-jumping rotation}
As an alternative, one may consider a more delay-optimized ferrying behavior which ensures that any existing message takes at most $\lceil \log N \rceil$ iterations to rotate from one cluster to another. The solution is to rotate all UAVs by $2^{\lceil \log N \rceil-1}$ clusters, then by $2^{\lceil \log N \rceil-2}$ and so on. Assuming good intra-cluster connectivity, this ensures that in each rotation (except potentially the last), the messages cached in any cluster are duplicated from one cluster to another that does not have the messages yet, leading to an optimal, exponentially fast spread of messages. The total distance travelled is then given by
\begin{align} \label{eq:d2}
    D_{\mathrm{tot}} = \langle \mathbf D, \hat{\mathbf W} \rangle_F = \operatorname{tr}( \mathbf D \hat{\mathbf W} ),
\end{align}
where $\hat{\mathbf W}$ is now the weight matrix parametrizing all jumps with entries $w_{ij} = 1$ whenever $j \in \{i + 2^k \mod N \mid k=1, \ldots, \lceil \log N \rceil \}$, resulting again in a QAP
\begin{align}
    \min_{\pi \in \mathrm{Sym}([N])} \operatorname{tr}( \mathbf D \mathbf P_\pi \hat{\mathbf W} \mathbf P_\pi^T).
\end{align}

The problem is similar in structure to the TSP and is likely hard, as its natural formulation is a QAP, which is known for being NP-hard in general and having few tractable special cases \cite{cela2013quadratic}. In particular, the matrix $\hat{\mathbf W}$ is a circulant matrix, and $\mathbf D$ is symmetric. However, $\mathbf D$ is neither anti-Monge nor Kalmanson, and hence does not fall into the few known, easy cases \cite{cela2013quadratic}. Since an analytic or exact solution is difficult as in TSP, we propose an approximate algorithmic solution. Though the problem is similar in structure to TSP, we are not aware of special algorithms for the considered QAP. We hence use a genetic algorithm (GA) \cite{mitchell1995genetic} with crossover operators used in TSP to solve for the optimal ordering of clusters. More precisely, we apply the well-known cycle crossover (CX) as e.g. found in \cite{umbarkar2015crossover}. The GA is given in Algorithm~\ref{alg:ga}.

\begin{algorithm}[b]
    \caption{GA for QAP $\min_{\pi} \operatorname{tr}( \mathbf D \mathbf P_\pi \hat{\mathbf W} \mathbf P_\pi^T)$.}
    \label{alg:ga}
    \begin{algorithmic}[1]
        \STATE Initialize $k$ random permutations $\pi^{(0)}_1, \ldots, \pi^{(0)}_k$.
        \FOR {iterations $n=1, \ldots, N_{\mathrm{GA}}$}
            \STATE Select the $k/2$ best permutations $\{ \pi^{(n)}_i \}_{i=1,\ldots,k/2}$ with lowest cost $F(\pi^{(n-1)}_i) \coloneqq \operatorname{tr}( \mathbf D \mathbf P_{\pi^{(n-1)}_i} \hat{\mathbf W} \mathbf P_{\pi^{(n-1)}_i}^T)$.
            \STATE Crossover best permutations with prob. $p_c = 0.8$ (CX \cite{umbarkar2015crossover}) to obtain $k/2$ children $\{ \pi^{(n)}_i \}_{i=k/2 + 1,\ldots,k}$.
            \STATE Mutate children with prob. $p_m = 0.4$ by exchanging a random pair of entries in the permutation.
        \ENDFOR
        \RETURN Best permutation $\argmin_{\pi \in \{ \pi^{(n)}_i \}_{i=1,\ldots,k}} F(\pi)$
    \end{algorithmic}
\end{algorithm}
\section{Algorithmic properties} \label{sec:theo}
We briefly state some properties of our algorithms, with proofs in the appendix for readability.

\paragraph{Convergence of AM}
We know that the AM algorithm always converges, consistently resulting in good coverage.

\begin{proposition} \label{prop:1}
Algorithm~\ref{alg:am} monotonically converges in the cost \eqref{eq:Jcost}.
\end{proposition}

\paragraph{Impact of data ferrying on coverage}
In the static or quasi-static clustered regime, where UAVs are significantly faster than users, $v \gg v_{\mathrm{user}}$, we can obtain bounds on the time-averaged (QoS-free) coverage $\mathrm{\bar c} = -\int_0^\infty c(\mathbf {X}(t), \boldsymbol \xi(t)) \, \mathrm dt$.

\begin{proposition} \label{prop:2}
In the static case, for any $p_{\mathrm{cover}} > 0$, there exists $N > 0$ and $T_{\mathrm{rot}} > 0$ such that a fraction $p_{\mathrm{cover}}$ of time-averaged coverage can be achieved.
\end{proposition}

We also note that in case of at least twice the capacity necessary to cover all users, one could split UAVs into two equal sets, and rotate the locations of each set alternatingly such that one set of UAVs always provides coverage while the other provides DTN coverage. However, in this work, we focus on constrained scenarios where capacity is more limited. 

\paragraph{Bounds on data ferrying delivery time}
In the connected cluster case assuming full and instantaneous intra-cluster connectivity (i.e. messages are shared fully among participants of each cluster), we can also show bounds on the DTN delivery time for the TSP-based and binary-jumping.

\begin{proposition} \label{prop:3}
In the static connected cluster case, for any $T_{\mathrm{rot}} > \frac{\operatorname{diam}(\mathcal X)}{v}$, the maximum time-to-delivery of DTN messages does not exceed $2 T_{\mathrm{rot}} \lceil \log N \rceil$ for binary-jumping, and $2 T_{\mathrm{rot}} N$ for TSP-based rotation.
\end{proposition}

\paragraph{Time complexity}
The solutions we provided are scalable not only to many users, but also to many UAVs. For simplicity, with $R=1$ we have a theoretical complexity of $\mathcal O(\max(N,M)^3)$ with an empirical complexity of $\mathcal O(\max(N,M)^2)$ \cite{bonneel2011displacement} for computation of optimal transport solutions in Algorithm~\ref{alg:am}, and complexity $\mathcal O(N^2)$ for computing the medoids. Analogously, for data ferrying, though the TSP and binary-jumping QAP are hard in theory, in practice they can be solved as long as one can evaluate the fitness of the QAP formulation as a sum of $N \lceil \log N \rceil \in \mathcal O(N\log N)$ elements. Overall, our approach can be scaled up to large systems with many users and UAVs, which we also verify experimentally.

\section{Evaluation} \label{sec:eval}
In this section, we evaluate our algorithm, comparing between TSP-based rotation, binary-jumping, as well as some baselines. The evaluation parameters are found in Table~\ref{tab:params}.

\begin{table}[b!]
    \centering
    \caption{Default parameters for scenario, algorithm, DTN}
    \label{tab:params}
    \begin{tabular}{@{}ccc@{}}
    \toprule
    Symbol          & Name          & Value     \\ \midrule
    $\mathcal X$ & Area of operation & $[-\SI{4}{\kilo\metre}, \SI{4}{\kilo\metre}]^2$ \\
    $v$ & UAV velocity & $\SI{20}{\metre\per\second}$ \\
    $\Delta T$ & Time interval & $\SI{10}{\second}$ \\
    $T_{\mathrm{rot}}$ & Rotation interval & $6 \Delta T$ \\
    $T_{\mathrm{tot}}$ & Total time & $\SI{12}{\hour}$ \\
    $N$ & Number of agents & $10$ \\
    $M$ & Number of ground users & $100$ \\
    $R$ & Capacity ratio & $1.3$ \\
    $r$ & Cell coverage range & $\SI{1}{\kilo\metre}$ \\
    \midrule
    $N_{\mathrm{AM}}$ & AM iterations & $4$ \\
    $N_{\mathrm{GA}}$ & GA iterations & $100$ \\
    $k$ & GA generation size & $100$ \\
    \midrule
    $r_{\mathrm{DTN}}$ & DTN wireless range & $\SI{100}{\metre}$ \\
    $v_{\mathrm{DTN}}$ & DTN transmit speed & $\SI{6570}{\kilo\byte\per\second}$ \\
    $b_{\mathrm{DTN}}$ & DTN buffer size & $\SI{20}{\mega\byte}$ \\
    $T_{\mathrm{DTN}}$ & DTN message interval & $\SI{10}{\second}$ \\
    $m_{\mathrm{DTN}}$ & DTN message size & $\SI{25}{\kilo\byte}$ \\
    $t_{\mathrm{DTN}}$ & DTN TTL & $\SI{300}{\minute}$ \\
    \bottomrule
    \end{tabular}
\end{table}

\begin{figure}
    \centering
    \hfill{}
    \includegraphics[width=0.99\linewidth]{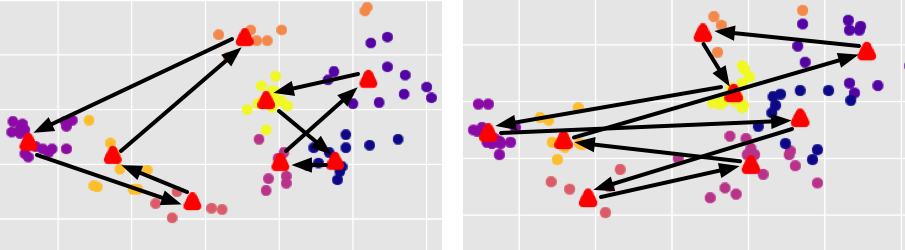}
    \hfill{}
    \caption{Snapshot as in Fig.~\ref{fig:example} for the SLAW model and binary-jumping, at two times. UAVs (red triangles) rotate through clusters of users (dots color-coded by cluster assignment).}
    \label{fig:slaw}
\end{figure}

\subsection{Scenarios} \label{sec:mobility}
For evaluation, apart from illustrative examples of Gaussian user clusters, both on a circle (Fig.~\ref{fig:modes}) and at city locations over a real city model (Fig.~\ref{fig:real}), we consider three advanced mobility models with a simulated duration of $\SI{43200}{\second}$~($\SI{12}{\hour}$). 
Scenarios contain $100$ nodes moving in an area of $\numproduct{8 x 8} \, \unit{\kilo\metre}$.
All movement traces, $10$ per mobility model, were generated using BonnMotion~\cite{aschenbruck2010bonnmotion} and the following configurations:

\paragraph{RWP}
As a simple baseline, we use a random waypoint model \cite{johnson1996dynamic} where nodes move with speeds between $\SI{1}{\metre\per\second}$ and $\SI{5}{\metre\per\second}$, with maximum pause time of $\SI{30}{\minute}$.

%minspeed=1.0
%maxspeed=5.0
%maxpause=1800.0

\paragraph{SMOOTH}
For a more realistic representation of human mobility, SMOOTH~\cite{munjal2011smooth} uses a power-law distribution of statistics with node clusters. We use $40$ clusters, an alpha for the flight distribution of $1.45$ (min: $1$, max: $14000$) and beta for the pause time distribution of $1.5$ (min: $10$, max: $1800$).

%range=100
%clusters=40
%alpha=1.45
%f_min=1
%f_max=14000
%beta=1.5
%p_min=10
%p_max=1800
%print=Not specified
%max_locations=20000

\paragraph{SLAW}
Another common model for simulating human mobility using a power-law distribution of statistics is SLAW~\cite{lee2009slaw}. Here, we use a pause time between $\SI{10}{\minute}$ and $\SI{50}{\minute}$, $500$ waypoints, a distribution weight of $3.0$, hurst of $0.75$ and a beta value of $1.0$. A snapshot is seen in Fig.~\ref{fig:slaw}.

%noOfWaypoints=500
%minpause=10.0
%maxpause=50.0
%beta=1.0
%hurst=0.75
%dist_weight=3.0
%cluster_range=50.0
%cluster_ratio=3
%waypoint_ratio=5

\subsection{Numerical results}
In the following, we demonstrate the results of our work on joint data ferrying and cell coverage by numerical simulations. Some qualitative examples are found in Fig.~\ref{fig:slaw} and Fig.~\ref{fig:real}.

\begin{figure}[b]
    \centering
    \hfill{}
    \includegraphics[width=0.99\linewidth]{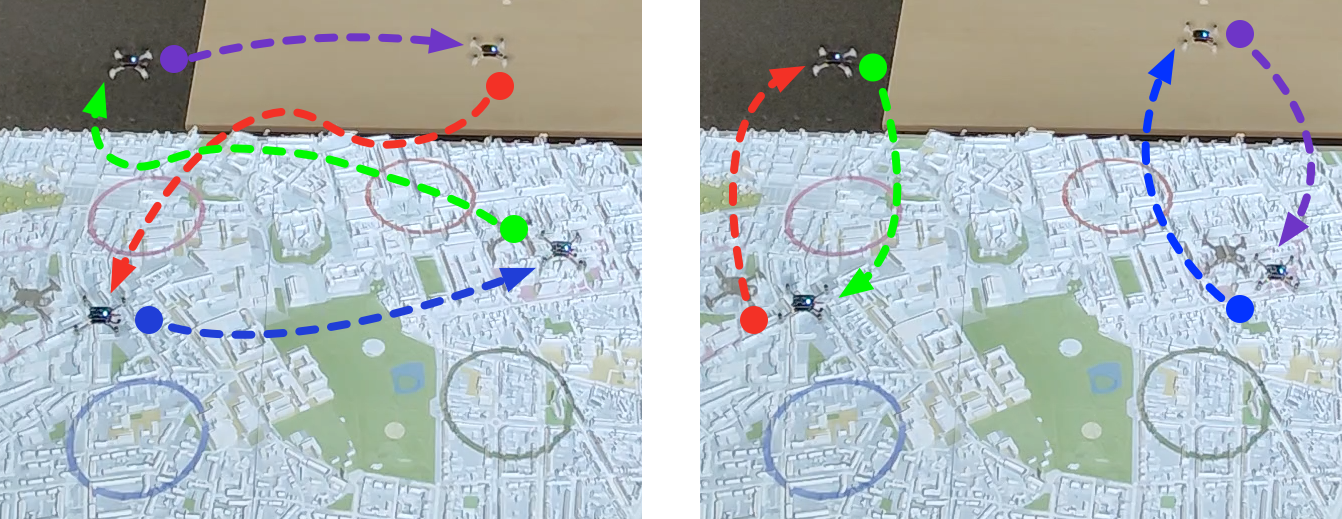}
    \hfill{}
    \caption{Real-world illustrative example implementation with four nano-quadrotors flying over a $1:1000$ city model with color-coded clusters of users on the ground, using artificial potential fields (APF, \cite{khatib1985real}) for collision avoidance. The colored dots mark previous positions of the UAVs, the dashed arrows show the flight trajectories according to binary-jumping. The UAVs spread messages between all clusters in $2$ jumps.}
    \label{fig:real} 
\end{figure}

\paragraph{Evaluation metrics}
We present the results of proposed approaches in comparison to reference schemes, considering: average time to delivery (TTD) for a message to reach its recipient; probability of delivery ($p_{deliver}$) for a message successfully reaching its recipient by end of scenario; time-averaged cell coverage ($\mathrm{\bar c}$); and total distance travelled by UAVs ($D_{tot}$). Since we are unaware of prior algorithms for joint cell coverage and DTN ferrying, as baseline we compare against the RWP heuristic, where UAVs repeatedly travel to a uniformly random point in $\mathcal X$, and the circular heuristic, where UAVs $i$ travel on a circle of radius $r_i = \SI{4}{\kilo\metre} \cdot (i / N)$ around zero. 

\paragraph{Quantitative evaluation}
Overall, in Table~\ref{tab:comparison} we see that our approaches are effective at providing both DTN message delivery and coverage, compared to the heuristics. We find that binary jumping performs best. Due to mobility of nodes and possibly disconnected clusters, in contrast to Proposition~\ref{prop:3} for the static connected case, TSP-based rotation improves over binary-jumping in TTD, though the confidence interval does not allow a definitive conclusion. Meanwhile, binary-jumping excels in travelled distance since less overall jumps mean less passing time between full rotations, for users to move and change the optimal distances \eqref{eq:d1}, \eqref{eq:d2} computed at the beginning of each rotation in the QAP. In Fig.~\ref{fig:theoretical_results}, we see that TTD is improved by using more UAVs, and distance travelled scales inversely with rotation time as expected. Lastly, in Fig.~\ref{fig:variations}, as predicted by Proposition~\ref{prop:2}, high coverage is achieved by tuning the number of UAVs or update interval. Simultaneously, delivery probabilities remain high.

\begin{table*}
    \centering
    % \vspace{0.06in}
    \caption{Performance comparison of proposed approaches and reference schemes, with $95\%$ confidence interval.}
    \label{tab:comparison}
    \begin{tabular}{@{}ccccc@{}}
    \toprule
    Algorithm          & Time to delivery TTD          & Delivery probability $p_{\mathrm{deliver}}$ & Average coverage $\mathrm{\bar c}$ & Distance travelled $D_{\mathrm{tot}}$     \\ \midrule
    TSP-based & $\SI{4146}{\second} \pm \SI{310}{\second}$ & $0.831 \pm 0.003$ & $0.77 \pm 0.03$ & $\SI{3055}{\kilo\metre} \pm \SI{164}{\kilo\metre}$ \\
    Binary-jumping & $\SI{4206}{\second} \pm \SI{305}{\second}$ & $0.831 \pm 0.003$ & $0.781 \pm 0.029$ & $\SI{2663}{\kilo\metre} \pm \SI{147}{\kilo\metre}$ \\
    Circular heuristic & $\SI{5088}{\second} \pm \SI{412}{\second}$ & $0.817 \pm 0.006$ & $0.495 \pm 0.014$ & $\SI{7233}{\kilo\metre} \pm \SI{10.3}{\kilo\metre}$ \\
    RWP heuristic & $\SI{4677}{\second} \pm \SI{359}{\second}$ & $0.825 \pm 0.005$ & $0.505 \pm 0.012$ & $\SI{8276}{\kilo\metre} \pm \SI{1.36}{\kilo\metre}$ \\
    No UAVs & $\SI{7124}{\second} \pm \SI{770}{\second}$ & $0.747 \pm 0.02$ & $-$ & $-$ \\
    \bottomrule
    \end{tabular}
\end{table*}

\begin{figure}
    \centering
    % \vspace{0.02in}
    \hfill{}
    \includegraphics[width=0.45\linewidth]{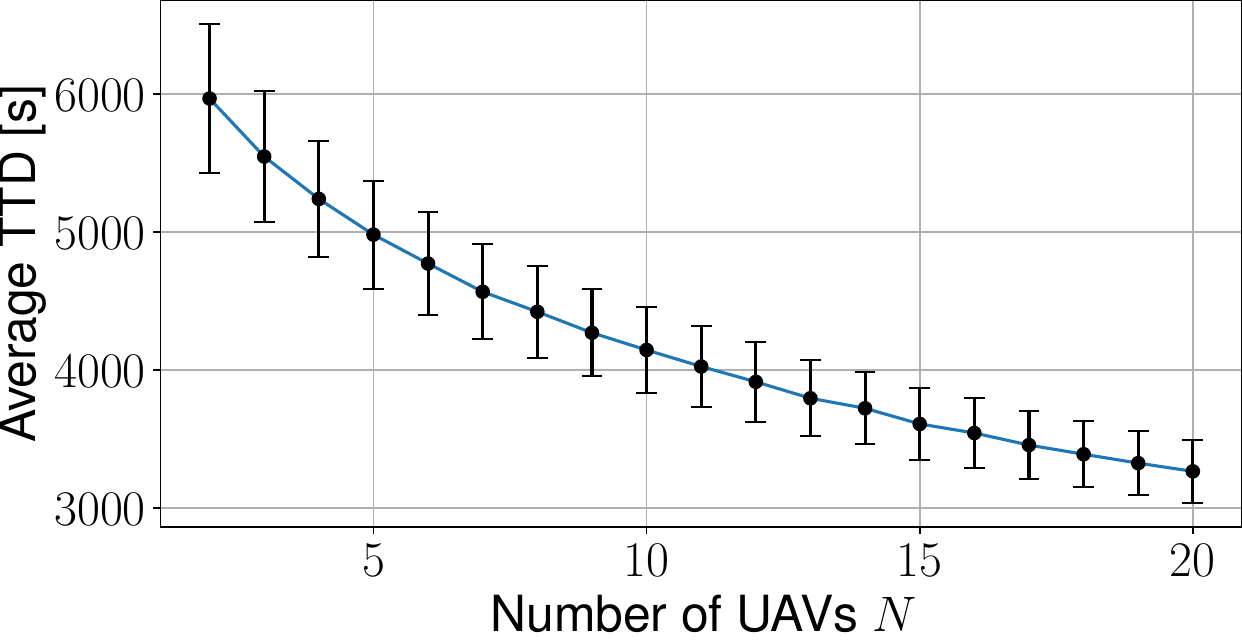}
    \hfill{}
    \includegraphics[width=0.45\linewidth]{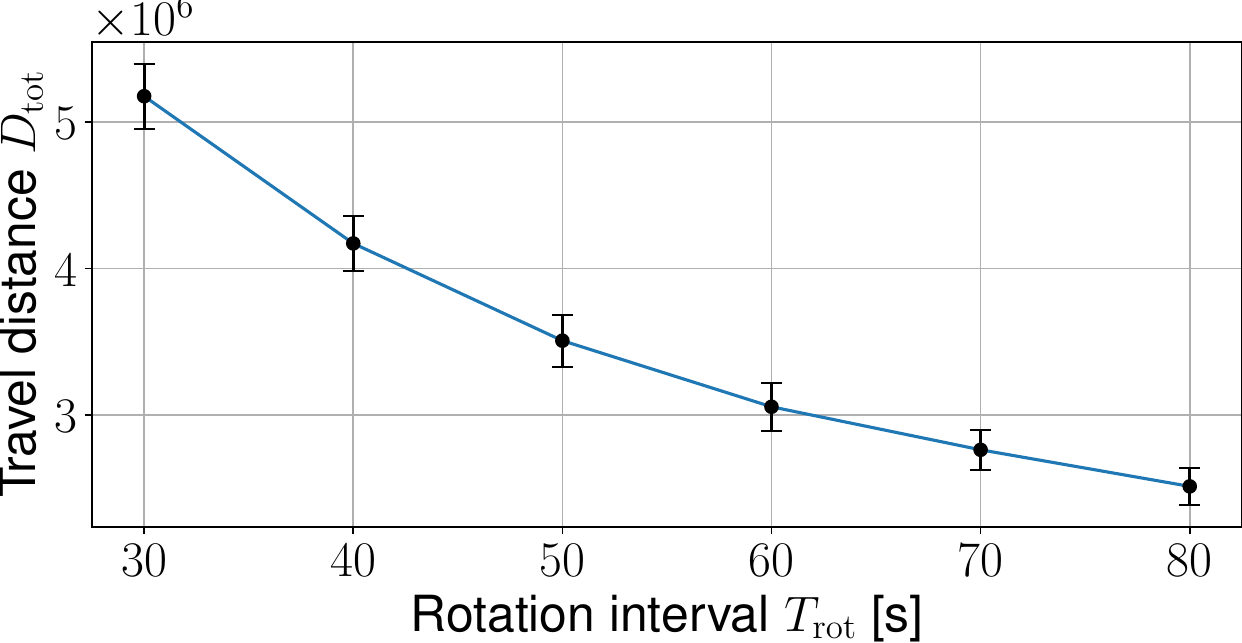}  
    \hfill{}
    \caption{The average TTD over number of UAVs and travelled distance over rotation interval for TSP-based rotation.}
    \label{fig:theoretical_results}
\end{figure}

\begin{figure}[b]
    \centering
    \begin{subfigure}[b]{0.45\linewidth}
    \includegraphics[width=\textwidth]{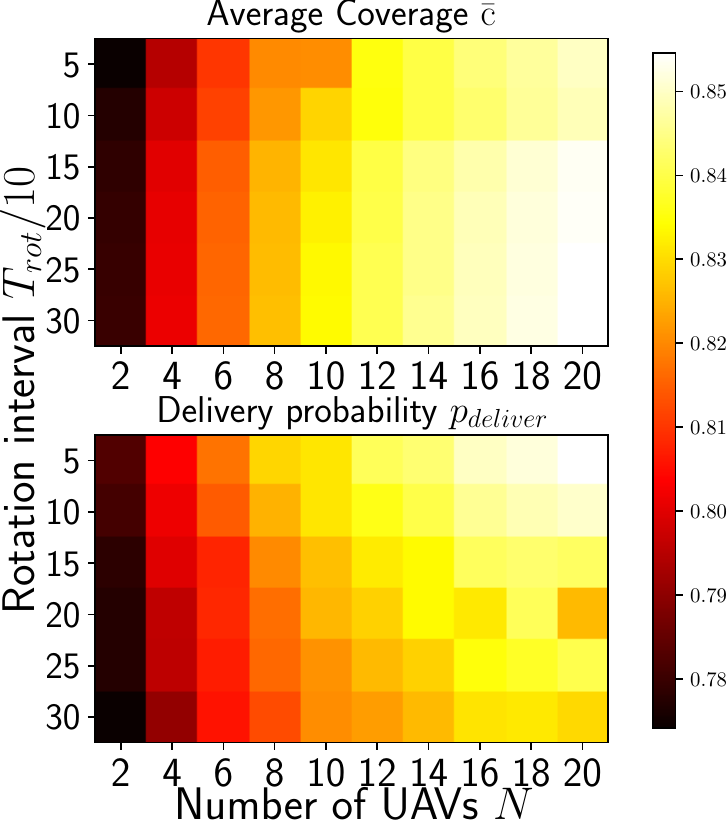}
    \caption{TSP-based}
    \end{subfigure}
    \begin{subfigure}[b]{0.45\linewidth}
    \includegraphics[width=\textwidth]{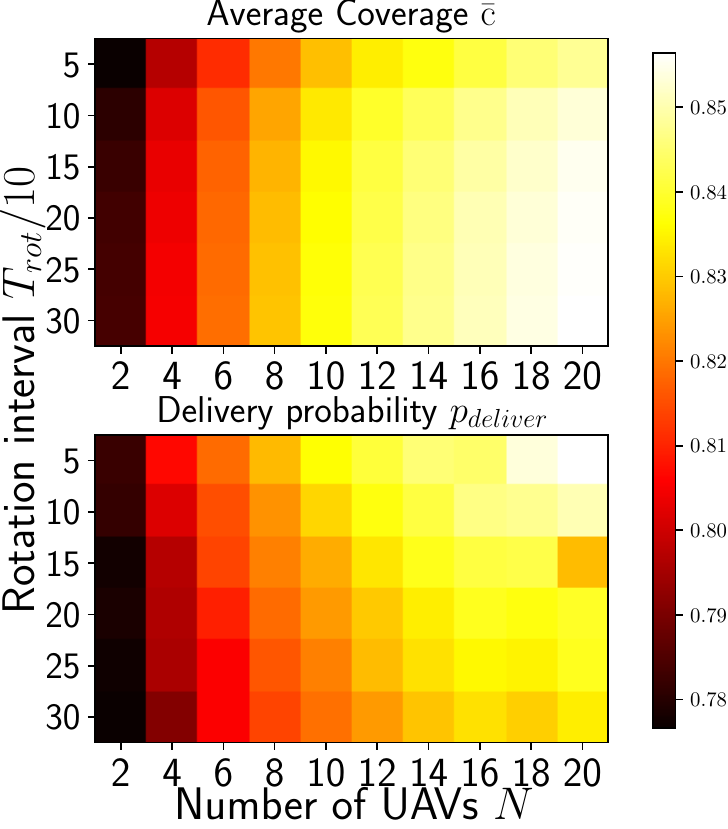}
    \caption{Binary-jumping}
    \end{subfigure}
    \caption{Average coverage and delivery probabilities across number of UAVs and rotation interval.}
    \label{fig:variations}
\end{figure}

\paragraph{Scalability of algorithms}
Finally, we verify the scalability of our proposed approach as discussed in Section~\ref{sec:theo}: In Fig.~\ref{fig:complexity}, the time to run Algorithm~\ref{alg:full} for $100$ epochs $n$ on a $\SI{2.4}{\giga\hertz}$ Quad-Core Intel Core i5 is shown for various numbers of ground users $M$ and UAVs $N$, and already remains feasible for real online scenarios (significantly below $100 \Delta T$). Nonetheless, the algorithm is not parallel and uses NumPy \cite{harris2020array} and POT \cite{flamary2021pot}, the run time of the algorithms scales close-to-linearly with the considered $N$ and $M$. Parallelization and code optimization may further improve the run time. 

\begin{figure}
    \centering
    % \vspace{0.02in}
    \includegraphics[width=0.99\linewidth]{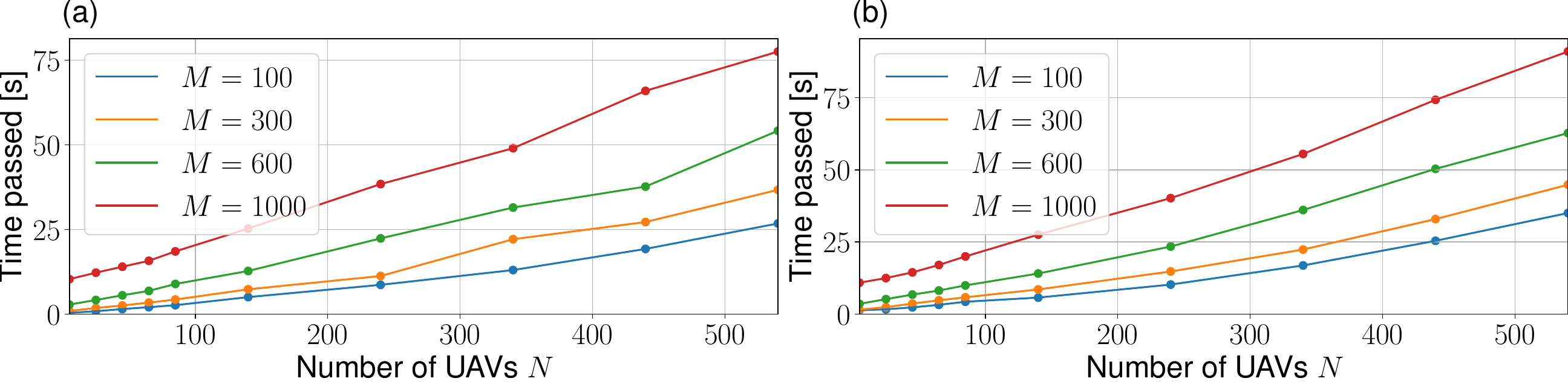}
    \caption{Time to run Algorithm~\ref{alg:full} for $100$ epochs $n$, for various numbers of ground users $M$ and UAVs $N$. The run time scales slowly with $N$ and $M$. (a): TSP-based; (b): binary-jumping.}
    \label{fig:complexity}
\end{figure}

\section{Related Work}
Existing literature has considered UAVs either for data ferrying using DTN or to provide flying base stations for cellular communications.
A vast literature have already been devoted to addressing aerial base stations \cite{mozaffari2017optimal,Zeng19proc}. More recent results have considered on the one hand, advanced physical layer communication techniques such as high-throughput relaying using buffer onboard UAVs \cite{Wang23icc}, directive transmission using high-gain directional antenna and gimbal-assisted adaptive antenna steering \cite{Lu22icc}, beamforming using aerial reconfigurable intelligent surface (RIS) \cite{Quispe21gs}, and exploiting multiple UAVs for cooperative multiple-input multiple-output communications \cite{Xiang20wcnc}. On the other hand, to fully exploit the benefits of these physical layer techniques, joint optimization of UAV deployment, trajectory planning and resource allocation have been considered for optimization of cell coverage \cite{Zhang21tcom}, throughput, energy \cite{ Xiang20wcnc, Lu22icc}, and mission completion time \cite{Wang23icc, Yilmaz22pimrc}. In addition to dynamic trajectory planning \cite{Wang23icc, Yilmaz22pimrc, Xiang20wcnc}, optimal selection of waypoints has been considered in \cite{Lu22icc} to fulfill requirements of e.g. UAV-aided monitoring and inspection, or urban airspace regulation. While UAV trajectory planning usually assumes knowledge of user locations, the impact of uncertain user locations was recently investigated in \cite{Yilmaz22pimrc}. However, other than few exceptions \cite{Wang23icc}, such literature usually does not consider DTN-based data ferrying. 

DTNs and UAVs for data ferrying were considered by other researchers which propose various different flight patterns.
Common approaches include covering areas with a triangle grid for a flight pattern, using genetic algorithms for path planning~\cite{harounabadi_evolutionary_2018}, distributed path planning algorithms mixed with task division~\cite{yoon_adaptive_2017} or specially consider in-transit nodes besides large clusters~\cite{zobel_topology-aware_2020}.
Considering multiple UAVs for data ferrying, Arafat et al.~\cite{arafat_location-aided_2018} propose a location-aided routing that exploits predictions of UAV movements for data forwarding.
Energy consumption is a major factor not only for data transmissions but also for the UAV flight time, thus, Zobel et al.~\cite{zobel_optimizing_2019} optimize inter-cluster data delivery based on communication performance and energy consumption of the UAV itself.
While Lieser et al.~\cite{lieser_simulation_2019} presented a networked UAV simulation framework with some basic flight patterns, similar to the other research lack a large scale comparison of baseline flight patterns as well as advanced ones.
Neither of the previously mentioned works combine data ferrying flight planning with secondary objectives such as providing cell services to ground nodes.

\section{Conclusion}
We have proposed and achieved usage of UAVs for joint DTN-based data ferrying and cellular coverage. Our approach is based on alternating minimization for optimal-transport-based cell coverage, QAP-based optimal data ferrying, and a combination of both in a synchronous algorithm. The approach was verified numerically and against simpler baselines. For future work, one could exploit the mobility of users, as the current approach remains agnostic to user behavior. Battery levels could be considered, e.g. by adding another cluster for recharging UAVs. Lastly, one could investigate the effects of unknown user locations, also by means of mean-field limits.

\section*{Appendix}
\begin{proof}[Proof of Proposition~\ref{prop:1}]
Algorithm~\ref{alg:am} induces a monotonically decreasing cost $J_{\mathbf X(t)}(\mathbf C^{(n)}, \boldsymbol \xi^{(n)})$, since 
\begin{align}
    J_{\mathbf X(t)}(\mathbf C^{(n+1)}, \boldsymbol \xi^{(n+1)}) &\leq J_{\mathbf X(t)}(\mathbf C^{(n+1)}, \boldsymbol \xi^{(n)}) \nonumber\\
    &\leq J_{\mathbf X(t)}(\mathbf C^{(n)}, \boldsymbol \xi^{(n)})
\end{align}
by \eqref{eq:alt1} and \eqref{eq:alt2} or \eqref{eq:alt2p}. By $J_{\mathbf X(t)}(\mathbf C, \boldsymbol \xi) \geq 0$ and monotone convergence, $(J_{\mathbf X(t)}(\mathbf C^{(n)}, \boldsymbol \xi^{(n)}))_{n \in \mathbb N}$ converges.
\end{proof}

\begin{proof}[Proof of Proposition~\ref{prop:2}]
First, there trivially exists $N > 0$ such that the achievable coverage is at least $\frac{1 + p_{\mathrm{cover}}}{2}$, e.g. using $N=M$ and singleton clusters (in practice, $N$ may be lower). As it takes at most $\frac{\operatorname{diam}(\mathcal X)}{v}$ time to reach a new cluster, choose $T_{\mathrm{rot}} = \frac{\operatorname{diam}(\mathcal X)}{v} \cdot k$ for any $k > 1$ to achieve a fraction $\frac{1 + p_{\mathrm{cover}}}{2} \cdot \frac{k-1}{k}$ of time-averaged coverage (as the fraction of time $\frac{\operatorname{diam}(\mathcal X)}{v} \cdot (k-1)$ is spent stationary). Choosing $k$ sufficiently large gives $\frac{1 + p_{\mathrm{cover}}}{2} \cdot \frac{k-1}{k} > p_{\mathrm{cover}}$ as desired.
\end{proof}

\begin{proof}[Proof of Proposition~\ref{prop:3}]
For binary-jumping, it takes at most $T_{\mathrm{rot}} \lceil \log N \rceil$ time until the next round of updates, after a message spawns in cluster $i$. During the new round, the message spreads to all clusters, since for any cluster $j$, there exist $k_1, \ldots, k_{\lceil \log N \rceil} \in \{0, 1\}$ with $\sum_{l=1}^{\lceil \log N \rceil} 2^{k_l} = j-i \mod N$ by binary representation of $j-i \mod N$, and hence there is a route of jumps performed by (possibly different) UAVs from $i$ to $j$. A similar argument holds for TSP-based rotation.
\end{proof}

\bibliographystyle{IEEEtran}
\bibliography{references}

\end{document}